\documentclass[final, 5p, times, twocolumn]{elsarticle}
\usepackage{lineno}
\usepackage{times}
\usepackage{graphicx,subfigure,epsfig,overpic,color}
\usepackage{url,amsmath,amssymb,amsthm,bm}
\usepackage{multirow,proof}
\usepackage{threeparttable,array,tabularx}

\usepackage[ruled,boxed]{algorithm2e}
\usepackage{natbib}

\journal{X}

\biboptions{authoryear}

\begin{document}

\begin{frontmatter}

\title{Bundle Fragments into a Whole: Mining More Complete Clusters via Submodular Selection of Interesting webpages for Web Topic Detection}

\author[a]{Junbiao Pang}
\ead{junbiao\_pang@bjut.edu.cn}

\author[a]{Anjing~Hu}
\ead{asuka004@emails.bjut.edu.cn}

\author[a,b]{Qingming~Huang}
\ead{qmhuang@ucas.ac.cn}

\address[a]{Faculty of Information Technology, Beijing University of Technology, No.100 Pingleyuan Road, Chaoyang District, Beijing 100124, China}

\address[b]{School of Computer and Control Engineering, University of Chinese Academy of Sciences,\\ No.19 Yuquan Road, Shijingshan District, Beijing 100049, China}

\begin{abstract}

Organizing interesting webpages into hot topics is one of key steps to understand the trends of multimodal web data. A state-of-the-art solution is firstly to organize webpages into a large volume of multi-granularity topic candidates; hot topics are further identified by estimating their interestingness. However, these topic candidates contain a large number of fragments of hot topics due to both the inefficient feature representations and the unsupervised topic generation. This paper proposes a bundling-refining approach to mine more complete hot topics from fragments. Concretely, the bundling step organizes the fragment topics into coarse topics; next, the refining step proposes a submodular-based method to refine coarse topics in a scalable approach. The propose unconventional method is simple, yet powerful by leveraging submodular optimization, our approach outperforms the traditional ranking methods which involve the careful design and complex steps. Extensive experiments demonstrate that the proposed approach surpasses the state-of-the-art method (\textit{i.e.}, latent Poisson deconvolution~\cite{pang-tao-lpd-icme-2016}) 20\% accuracy and 10\% one on two public data sets, respectively.

\end{abstract}


\end{frontmatter}


\section{Introduction}\label{sec:intro}

With the boom of web technologies, recently it is convenient for people to share data on social media websites which successfully facilitates both the generation and the propagation of User Generated Content (UGC)~\cite{pang2013unsupervised}. Users are constantly struggling to keep
up with the ever-increasing amounts of content that is
being published every day. Extracting the hot yet interesting topics from a sea of webpages is an efficient approach to handle the \emph{information overload} problem.

Hot topics provide an easy way to abstract the public concerns, social trends and urgent events.
However, the unprecedented explosion of UGC data has made it difficult for web users to quickly access hot web topics~\cite{bakshy-messing-adamic-science2015}. Driven by such pressing requirement, topic detection from web~\cite{Cao-Ngo-Zhang-Li-2011}~\cite{pang-tao-lpd-icme-2016}~\cite{Zhang-Li-Chu-Wang-Zhang-Huang-2013}~\cite{pang-tao-nucom-describe-topic-2017} is such an effort to discover a tiny fraction of interesting webpages strongly connected by a seminal event from a sea of social media~\cite{pang2013unsupervised}. We believe that the ability to remove the uninteresting webpages yet organize the interesting ones into different seminal events lies at the basis of understanding trends.

\newtheorem{myobr}{Observation}
\newtheorem{mythe}{Theorem}
\newtheorem{mypro}{Proposition}
\newtheorem{mydef}{Definition}
\newtheorem{myremark}{Remark}

Discovering hot web topics is like looking for a needle in a haystack. Because only a small number of the interesting webpages could evolve into the hot topics eventually. To handle this challenge, the Detection-By-Ranking (DBR) approach~\cite{pang2013unsupervised} is proposed. Firstly, a huge number of the multi-granularity topic candidates is generated, attempting to hit all hot topics in a brute-force approach; secondly, the interestingness of the topics is estimated to discover hot topics. However, due to the lack of the supervised information, the multi-granularity approach inevitably generates the fragments of hot topics~\cite{kamnitsas-SSL-imcl-18} which naturally lead to the inferior performances~\cite{pang2013unsupervised}~\cite{pang-tao-lpd-icme-2016}.
Because Poisson Deconvolution (PD)~\cite{pang2013unsupervised} empirically tends to rank fragments of hot topics ahead. The essence of this problem is originally from the unsupervised approach.

A naive idea attempts to generate complete topics in presence of noises via advantaged clustering methods, \emph{e.g.}, low-rank models~\cite{lin-aimi2016}. However, a sea of uninteresting webpages would make the advantaged methods ineffective. For instance, the noise robust spectral clustering~\cite{aleksander-yves-kdd17} is ineffective to find clusterings when about 96\% noises are presented. Moreover, both the multi-language
words and the user-defined abbreviations make the dimension of the features from textual modality very high. For instance, Term Frequency-Inverse Document Frequency
(TF-IDF) from both Chinese and English has about 54,000 dimensions, which make the univariate projection based method~\cite{Maurus-kdd16} fail.

In this paper, we seek a \emph{fragment-to-whole} and \emph{coarse-to-fine} method, based on four motivations. Firstly, bundling the fragment topics into a whole would produce complete yet coarse topics; this avoids the drawback of the NMS~\cite{hosang-learning-nms-17}~\cite{he-uncertainty-loss-nms-19}. Secondly, refining a topic into the accurate, coherent and interpretable~\cite{pang-tao-cpd-tnnls-19} one is expected. Thirdly, we avoid the open problem of discovering clusterings in a sea of noises~\cite{Maurus-kdd16}. Fourthly, without time-complexity optimization, a scalable approach is desired. In summary, we avoid the open problem of clustering in a sea of noises by putting fragment topics into refined ones in a scalable approach.

Therefore, we propose the Bundling-Refining (BR) method as follows:
 \begin{itemize}
\item \emph{Bundling:} This step agglomerates topics into a coarse one (more details will be discussed in Section~\ref{sec:subsect:bundling}). This step produces a over-complete topics from the fragments, and reduces the number of false detected topics at same time.
\item \emph{Refining:} This step models the interestingness of webpages, and then removes uninteresting webpages in a scalable approach. Concretely, a over-complete topic is firstly represented as a graph, where the interestingness of a webpage is modeled as the stationary distribution.
\end{itemize}

To the best of our knowledge, this is the first to mine a whole topic from fragments for topic detection on web. The proposed method is computationally simple, and yet exceptionally powerful. Simply by greedily selecting a set of webpages, without complex training process, we find a novel clustering method that exceeds or meets the traditional approach~\cite{pang2013unsupervised}~\cite{pang-tao-lpd-icme-2016} on two public data sets.

The rest of this paper is organized as follows: Section~\ref{sec:relatedwork} reviews the related work. We describe the details of our approach in Section~\ref{sec:lsPD}. Experimental results are presented in Section~\ref{sec:results}, and the paper is concluded in Section~\ref{sec:conclusion}.

\section{Related Work}\label{sec:relatedwork}

\textbf{Web Topic Detection.} Webpages are the typical heterogeneous data. As a result, many literatures have considered topics as clusterings from the multi-modal data~\cite{zhang2024kbs}~\cite{seema2023cmtd}~\cite{PARK2023esa-hottext-topic}. There are two important research threads. One adapts a single-modal clustering algorithm to the multi-modality~\cite{liu2018esa-text-image}, and the other is the similarity graph method~\cite{papadopoulos-11}, where the multi-modalities data are fused into the edges of a graph.

In the former case, topic detection extends the single-modality based methods into multi-modal data. For example, multi-modal LDA~\cite{liu2018esa-text-image} was proposed to group images with tags into topics. In the similarity graph method, multi-modal cues were fused into edges of a similarity graph. For instance, Cao \emph{et al.}~\cite{Cao-Ngo-Zhang-Li-2011} first generate events on video tags via clustering, and then link these clusterings into topics by the textual-visual similarity. As a comparison, the fused similarity graph was computationally simple and easily extendable to the other graph-based algorithm~\cite{pang-hu-interretation-topic-TCY-19}.

The DBR method~\cite{pang-tao-lpd-icme-2016}~\cite{pang-tao-nucom-describe-topic-2017} belonged a similarity graph approach. Although a large number of methods has been proposed to clustering topics~\cite{yang2012clustering}~\cite{Zhang-Li-Chu-Wang-Zhang-Huang-2013}, a few work notices that many clusters are the fragments of hot topics.

Classical topic models have been proposed to infer hidden themes for document analysis, including Latent Dirichlet Allocation (LDA)~\cite{blei-jmlr-03}, probabilistic Latent Semantic Analysis (pLSA)~\cite{hofmann-plsa-99} and various variations. These topic models generally work well on long and structured documents~\cite{baldwin-CoNLL-2009}. Moreover, these models have assumed that each webpage must belong to a topic. For instance, a subset of twitters was firstly retrieved by a set of predefined keywords, and then LDA was used to organize these twitters into topics~\cite{Shi-kan-www-short-text-18}. On the contrary, our task is to mine a few hot topics in a sea of noise webpages~\cite{qian-zhang-xu-mm16}. {\color[RGB]{0,0,255} Recently, ~\cite{wang-2023-text2topic-EmNLP} have used the pre-trained models based on Transformer~\cite{vaswani-2017-transformer-nips} obtain the zero-shot text embedding feature to classifier topic.  }

\textbf{Clustering In a Sea of Noise.} The ideal approach for hot topics detection from web was to cluster in a sea of noise~\cite{Maurus-kdd16}. We also empirically find that only a small fraction of interesting webpages (about 3\%) have evolved into hot topics~\cite{lin2022-cfskd-conf}. Obviously, the partition-based methods are impossible to deal with this problem since this approach assumes that every data point must belong to a cluster~\cite{ACM-data-density-23}.

There has been a trend of introducing noise-resistant clustering techniques into web topic detection. For instance, Li et al~\cite{li-joo-qi-zhu-tmm2016} utilized Swendsen-Wang cuts to group topics, and Zhang et al~\cite{Zhang-Li-Chu-Wang-Zhang-Huang-2013} used the graph shift to group a few webpages into hot topics.

In fact, web topics are not equivalent to conventional clusters or the zero-shot topic classification ~\cite{wang-2023-text2topic-EmNLP} since the web topics faces the more semantically weak feature representation than that of other scenarios, \emph{e.g.}, high-dimension data~\cite{pandove-goel-rani-review-high-dimension-18}, budgeted optimization~\cite{byrka-pensyl-rybicki-srinivasan-trinh-budgeted-17}.
These clustering methods naturally produce fragments of hot topics.

\section{The Proposed Bundling-Refining Method}\label{sec:lsPD}

In this paper, for a set of webpages $\{\mathbf{x}_1,\ldots,\mathbf{x}_N\}$, the visual feature $f^{vis}$ and the textual feature $f^{txt}$ are firstly extracted from each webpage. Two similarity matrices $W^{vis}$ and $W^{txt}$ are constructed from the visual feature and the textual one, respectively\footnote{In the following, we ignore either the subscript or the superscript in the different context, if it does not cause any confusion.}.

\subsection{Ranking Multi-granularity Topic Candidates}

Following the approach in~\cite{pang-tao-nucom-describe-topic-2017}, we convert the matrix $W$ into the affinity one $\hat{W}$ by Gaussian kernel, \emph{i.e.}, $\hat{W}_{ij}=\exp\left(-\|W_{ij}\|_2^2/\sigma^2\right)$ where $\|\cdot\|_2$ denotes the $\ell_2$ norm, and $\sigma^2$ denotes the deviation. Once these affinity matrices are constructed, we build the visual and the textual $k$-Nearest Neighbor Graphs ($k$-N$^2$Gs), {\color[RGB]{0,0,255} where the $k$ nearest neighbors of each webpage truncate the affinity matrix $\hat{W}$ into the spare one $A$ as follows}:
\begin{eqnarray}\label{eqt:knngraph}
G^{vis}&=(V,E^{vis},A^{vis}),\\
G^{txt}&=(V,E^{txt},A^{txt}),
\end{eqnarray}
where vertex $V$ denotes webpages, edge $E^{vis}$ ($E^{txt}$) represents whether two webpages are connected by the visual (textual) modality or not, and matrix $A^{vis}$ ($A^{txt}$) describes the closeness of two webpages determined by the visual (textual) features. Subsequently, a mixed graph is obtained as follows:
\begin{equation}\label{eqt:hybridgraph}
G=\left(V,E, A \right),
\end{equation}
where $E=E^{vis}\cup E^{txt}$, and $A=(A^{vis}+A^{txt})/{2}$. {\color[RGB]{0,0,255} $E=E^{vis}\cup E^{txt}$ means that either the individual modality or the multi-modalities could be used to model the similarity between two webpages.}

By Similarity Cascade (SC)~\cite{pang2013unsupervised}, the multi-granularity topics $C_k$ ($k=1,\ldots,K$) are generated from the mixed graph $G$. A topic $C_k$ is formally represented as follows:
\begin{equation}
C_{k}=c^{\top}_k{\circ} c_k,
\end{equation}
in which the indicator vector $ c_k \in \{0, 1\}^{1\times N}$, where $1$ or $0$ means that whether the topic $C_k$ contains the webpage $\mathbf{x}_i$ or not. The operation $\circ$ means that the diagonal of the matrix $c^{\top}_k c_k$ is set to zero. The weight of a topic $\mu_k$ is estimated as follows:
\begin{equation} \label{eqt:poissondeconvolution}
\begin{split}
w_{ij}&\sim \text{Poisson}(a_{ij})\\
s.t.: & \  \ w_{ij} = \sum_{k=1}^K \mu_k C_{k_{ij}}.
\end{split}
\end{equation}
The interestingness of a topic is estimated as \mbox{$i_k=\mu_k\cdot |C_k|$}, where $|C_k|$ is the number of webpages in a topic $C_k$. Therefore, we obtain the interestingness order of topics \mbox{$\mathcal{L}=i_1\prec \ldots i_k \ldots \prec i_K$}. For more details about the PD-based ranking, please refer to~\cite{pang2013unsupervised}.

\subsection{Bundling Fragments into a Coarse Topic}\label{sec:subsect:bundling}

We assume that every topic is a fragment of some hot topic. That is, each topic should at least be bundled into a coarse one. Given the order list $\mathcal{L}$ and the bundling window $W$, Jaccard distance between the \mbox{$k$-th} topic and the subsequent one \mbox{$j$-th} ($j \in\{ k{+}1,\ldots,k{+}W\}$) is computed as,
\begin{equation}\label{eqt:bundle}
J_{kj}= \frac{|C_k \bigcap C_j|}{|C_k \bigcup C_j|},
\end{equation}
where $\bigcap$ and $\bigcup$ respectively represent the intersection operation and the union operation, and $|\cdot|$ denotes the number of webpages in a set. If $J_{kj}$ is large than a predefined threshold $\tau$, we bundle the $j$-th topic and the $k$-th one into a coarse topic, \emph{i.e.}, $\widehat{C}_k=\bigcup C_i$ ($i\in\{k,\ldots, k{+}W\}$). Next, we replace the topic $C_k$ with the bundled one $\widehat{C}_k$, as illustrated in Fig.~\ref{fig:bundling}.

\begin{figure}[t!]
\centering
\includegraphics[width=.33\textwidth]{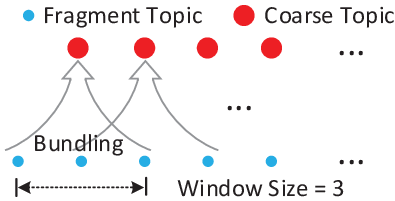}
\caption{Bundling fragment topics into a coarse one. In this example, the size of the bundling window is 3.}
\label{fig:bundling}
\end{figure}

Theoretically, the bundling process has a time complexity of $O(K{\cdot}W)$ ($W{\ll} K$). In practice, $W$ is frequently smaller than 50. Therefore, the bundling process is quite efficient. However, the uninteresting webpages in the fragments are inevitably merged into the coarse topics. In this paper, the interestingness of webpages is estimated in Section~\ref{sec:sub:interestingness}, and a submodular-based approach is proposed to refine coarse topics in Section~\ref{sec:sub:refinging}.

\subsection{Modeling Interestingness by Stationary Distribution}\label{sec:sub:interestingness}

An interesting webpage is defined as the likelihood of that people are attracted. If the correlations among webpages are encoded as a graph, the stationary distribution measures the interestingness of a webpage~\cite{page-brin-motwani-pagerank-99}. However, the structure information between webpages (\emph{e.g.}, link) is unavailable. We convert a coarse topic into a graph with the similarities among webpages.

The reconstructed similarity between webpages is computed as follows:
\begin{equation}
\widehat{S}_{ij} = \sum_{C_k\in \widehat{C}}\mu_kC_{k_{ij}},
\end{equation}
where $\widehat{S}_{ij}$ ($\widehat{S}_{ij}\in \mathbb{R}^{|\widehat{C}|\times |\widehat{C}|}$) denotes the point-wise similarities between the webpages from the coarse topic $\widehat{C}$. Rather than using the original similarity $A$, $\widehat{S}$ is considered as a refined similarity~\cite{pang-tao-nucom-describe-topic-2017}.

Consequently, a graph \mbox{$\widehat{G}=(\widehat{V},\widehat{E},\widehat{S})$} is constructed, where node $\widehat{V}$ represents the webpages in the coarse topic $\widehat{C}$, and edge $\widehat{E}$ denotes the connectivity between these webpages, \emph{i.e.}, if~$\widehat{S}_{ij}>0$, \mbox{$\widehat{E}_{ij}=1$}; otherwise, $\widehat{E}_{ij}=0$.

Assuming a random walker travelling on the given graph $\widehat{G}$, the probability of a transition from node $i$ to node $j$ is $P_{ij}=\widehat{S}_{ij}/d_i$, where $d_i=\sum_j\widehat{S}_{ij}$ is the out-degree of the node $i$. $P_{ij}$ is the probability of the node $i$ being visited by the random walker at the $0$-th step.

Aiming at computing the stationary distribution $\pi$, random jumps~\cite{kleinberg-random-walk-ACM-99} are introduced into $P_{ij}$ as follows:
\begin{equation}\label{eqt:stationarydistribution}
\pi_j=\alpha\sum_{i\in \widehat{V}}\pi_i P_{ij} + (1-\alpha)\frac{1}{|\widehat{C}|},
\end{equation}
where $\alpha\in (0,1)$ denotes the decay factor that respects the probability with which a walker follows the connected edges, and $1{-}\alpha$ is the probability of jumping to a random node.~\eqref{eqt:stationarydistribution} is PageRank~\cite{page-brin-motwani-pagerank-99} algorithm, where $\alpha$ is usually recommended to range from 0.8 to 0.95~\cite{meyer-langville-pagerank-04}. In this paper, we set $\alpha = 0.9$ in~\eqref{eqt:stationarydistribution}.

By defining the new transition matrix as $\widehat{P}\leftarrow \alpha P+(1-\alpha)\mathbf{e}\mathbf{e}^\top/|\widehat{C}|$
(where the vector $\mathbf{e}$ is a column of ones), the stationary distribution can be efficiently solved by the power method~\cite{meyer-langville-pagerank-04} with the time complexity $O(T\cdot nz(P))$, where $T$ is number of iterations in the power method, and $nz(P)$ is number of non-zero elements in $P$. {\color[RGB]{0,0,255} In our experiment, the size of coarse topics $\widehat{C}$ range from 20 to 93. Therefore, the computation cost of~\eqref{eqt:stationarydistribution} is very small.}


\subsection{Refining a Coarse Topic by Submodular Selection}\label{sec:sub:refinging}

To remove uninteresting webpages from a coarse topic, we establish two criteria: 1) \textbf{Interestingness:} a hot topic consists of the interesting webpages. 2) \textbf{Similarity:} the interesting webpages in a hot topic should be more similar than that of the uninteresting webpages. By following the above criteria, an immediate solution is to select the interesting and similar webpages as follows:
\begin{equation}\label{eqt:naivemethod}
\begin{split}
\arg & \max_{I}  \sum_{i,j}(I_i\pi_i) \widehat{S}_{ij}(I_j\pi_j) \\
s.t.:& I\mathbf{e}^\top=k,
\end{split}
\end{equation}
where $I\in\{0,1\}^{1\times |\widehat{C}|}$ is the indicator vector, $\widehat{S}$ is the similarity matrix, and $k$ is the number of the selected webpages. However, the optimization of~\eqref{eqt:naivemethod} is equal to quadratic integer programming which is NP-hard problem.

In order to approximate~\eqref{eqt:naivemethod}, our approach selects a subset of webpages $\mathcal{P}$ as follows:
\begin{equation}\label{eq:goodnessFunction}
\arg \max_{\mathcal{P}} g(\mathcal{P} ) =\lambda\sum_{i \in \mathcal{P}}\pi_i - \sum_{i,j\in \mathcal{P}}\pi_i D_{ij} \pi_j,
\end{equation}
where $i$ is the index of the selected webpage $\mathbf{p}_i\in \mathcal{P}$, $D_{ij}$ is the dissimilarity between two webpages, and $\lambda>0$ is a parameter that trades off between the interestingness and the dissimilarity.

In~\eqref{eq:goodnessFunction}, the first term means that the selected webpages should have the high interestingness scores; while, the second term requires that the selected webpages should have small dissimilarity values between each other. Therefore, the definition of the dissimilarity $D_{ij}$ should follow the following property:
\begin{mydef}{[Property of $D_{ij}$ in~\eqref{eq:goodnessFunction}]}
Given three webpages $\mathbf{x}_i$, $\mathbf{x}_j$, and $\mathbf{x}_k$, as well as the corresponding dissimilarity values $D_{ij}$ and $D_{kj}$ between these webpages, if $\mathbf{x}_i$ is more similar to $\mathbf{x}_j$ than that of between $\mathbf{x}_k$ and $\mathbf{x}_j$, then dissimilarity values should satisfy that $D_{ij} < D_{kj}$.
\end{mydef}

By the definition of the dissimilarity $D$ in~\eqref{eq:goodnessFunction}, the goodness function $g(\mathcal{P} )$ approximates the objective function~\eqref{eqt:naivemethod}. In this work, considering that $\widehat{S}_{ij}\in (0,1)$, the dissimilarity $D$ is converted from $\widehat{S}$ by Gaussian kernel, \emph{i.e.}, \mbox{$D_{ij}= \text{exp}(-\widehat{S}^2_{ij}/ \sigma)$}, where $\sigma$ is the bandwidth to control the decay speed of the similarity $\widehat{S}$. In this paper, we set $\sigma= 10$.

\subsubsection{Submodular Selection of~\eqref{eq:goodnessFunction}}

We present the so-called diminishing return property of the goodness function~\eqref{eq:goodnessFunction} in Proposition~\ref{pro:diminishing}. By Proposition~\ref{pro:diminishing}, if we add more webpages into an existing subset $\mathcal{P}$, the goodness of~\eqref{eq:goodnessFunction} is non-decreasing.

\begin{mypro}\label{pro:diminishing}{[Diminishing of $g(\mathcal{P})$ in~\eqref{eq:goodnessFunction}]}
The goodness function in~\eqref{eq:goodnessFunction} has the following properties:
\begin{itemize}
\item (\text{P1: Submodularity.}) for any $\lambda>0$, the goodness $g(\mathcal{P})$ is submodular with respect to $\mathcal{P}$;
\item (\text{P2: Monotonicity.}) for any $\lambda\geq2$ and $\sum_{ij}D_{ij}=1$, the goodness $g(\mathcal{P})$ is monotonically non-decreasing with respect to $\mathcal {P}$.
\end{itemize}
\end{mypro}

\begin{proof}
We firstly prove (P1). For any $\mathcal{P}_1 \subset \mathcal{P}_2$ and a webpage $\mathbf{x} \notin \mathcal{P}_2$, we have
\begin{smaller}
\begin{equation}\label{eqt:dimensionP1}
\begin{split}
g(\mathcal{P}_1 \cup \mathbf{x}) {-} g(\mathcal{P}_1) &= \lambda \sum_{i \in (\mathcal{P}_1 \cup \mathbf{x} )} \pi_i - \sum_{(i,j) \in (\mathcal{P}_1 \cup \mathbf{x} )}\pi_i D_{ij} \pi_j\\
&-\left(\lambda \sum_{i \in \mathcal{P}_1} \pi_i - \sum_{(i,j) \in \mathcal{P}_1 }\pi_i D_{ij} \pi_j\right) \\
&= \lambda \pi_\mathbf{x} - \left( \sum_{i \in \mathcal{P}_1 }\pi_i D_{i\mathbf{x} } \pi_\mathbf{x} + \sum_{j \in \mathcal{P}_1 }\pi_\mathbf{x}  D_{\mathbf{x} j} \pi_j \right)
\end{split}
\end{equation}
\end{smaller}
Following~\eqref{eqt:dimensionP1}, we have $g\left(\mathcal{P}_2 \cup \mathbf{x} \right) - g\left(\mathcal{P}_2\right) =\lambda \pi_\mathbf{x} - \left( \sum_{i \in \mathcal{P}_2 }\pi_i D_{i\mathbf{x} } \pi_\mathbf{x}  + \sum_{j \in \mathcal{P}_2 }\pi_\mathbf{x}  D_{\mathbf{x} j} \pi_j \right)$.
Therefore, we have
\begin{equation*}
\begin{split}
&\left(g(\mathcal{P}_1 \cup \mathbf{x} ) - g(\mathcal{P}_1)\right)-\left(g(\mathcal{P}_2 \cup \mathbf{x} ) - g(\mathcal{P}_2)\right)\\
& = \sum_{i \in \mathcal{P}_2/\mathcal{P}_1 }\pi_i D_{i\mathbf{x} } \pi_\mathbf{x}  + \sum_{j \in \mathcal{P}_2/\mathcal{P}_1}\pi_\mathbf{x}  D_{\mathbf{x} j} \pi_j \geq 0,
\end{split}
\end{equation*}
which completes the proof of (P1).

Next, we prove (P2). Given any $\mathcal{P}_1\cap \mathcal{P}_2=\varnothing$, where $\varnothing$ means an empty set, we have
\begin{smaller}
\begin{equation}\label{eqt:diminisingP2}
\begin{split}
&g(\mathcal{P}_2\cup \mathcal{P}_1)-g(\mathcal{P}_2) =\lambda\sum_{i \in \mathcal{P}_1} \pi_i
-\left(\sum_{i,j\in \mathcal{P}_1}\pi_i D_{ij} \pi_j\right.\\
&+\sum_{i\in \mathcal{P}_1,j\in \mathcal{P}_2}\pi_i D_{ij} \pi_j+\left.\sum_{i\in \mathcal{P}_2, j\in \mathcal{P}_1}\pi_i D_{ij} \pi_j \right)\\
&=\lambda\sum_{i \in \mathcal{P}_1} \pi_i - \left(\sum_{i\in \mathcal{P}_1,j\in \mathcal{P}_1\cup\mathcal{P}_2 }\pi_i D_{ij} \pi_j
+ \sum_{i\in \mathcal{P}_2, j\in \mathcal{P}_1}\pi_i D_{ij} \pi_j \right)
\end{split}
\end{equation}
\end{smaller}
with the constraint $\lambda \geq 2$ and the inequality \mbox{$\sum a_ib_i\leq \sum a_i \sum b_i, (a_i\geq0, b_i\geq 0)$},~\eqref{eqt:diminisingP2} is relaxed as follows:
\begin{smaller}
\begin{equation*}
\geq \sum_{i\in \mathcal{P}_1}\pi_i\left(1 - \sum_{j\in \mathcal{P}_1\cup\mathcal{P}_2}D_{ij}\pi_j \right)+\sum_{j\in \mathcal{P}_1}\left(1 - \sum_{i\in \mathcal{P}_2}\pi_iD_{ij}\right)\pi_j
\end{equation*}
\end{smaller}
with the constraints $\sum_{ij}D_{ij}=1$, $\sum_i\pi_i=1$, the above inequality is:
\begin{equation*}
\geq \sum_{i\in \mathcal{P}_1}\pi_i\left(1 - \sum_{j\in \mathcal{P}_1\cup\mathcal{P}_2}\pi_j \right)+\sum_{j\in \mathcal{P}_1}\left(1 - \sum_{i\in \mathcal{P}_2}\pi_i\right)\pi_j\geq0.
\end{equation*}
which completes the proof of (P2).
\end{proof}

Based on the Proposition~\ref{pro:diminishing}, a subset of $k$ webpages is greedily selected by the following discrete derivative of $g(\mathcal{P})$:
\begin{equation}\label{eqt:discrete-discrete}
\begin{split}
\triangle(\textbf{p}|\mathcal{P}) &= g(\mathcal{P}\cup\{\textbf{p}\}) - g(\mathcal{P})\\
&=\lambda \pi_\textbf{p}- \left( \sum_{i \in \mathcal{P}}\pi_i D_{i\textbf{p}} \pi_\textbf{p} + \sum_{j \in \mathcal{P}}\pi_\textbf{p} D_{\textbf{p}j} \pi_j \right).
\end{split}
\end{equation}
In this paper, $\lambda$ is assigned as $2$. Alg.~\ref{alg:optimize-goodness} provides a greedy solution to adaptively refine topics with a $(1-1/e)$ near-optimal solution~\cite{nemhauser-wolsey-fisher-submodular-73}. The time complexity of the submodular selection in Alg.~\ref{alg:optimize-goodness} is $O(|\widehat{C}|\cdot M)$ where $M$ is the number of the webpages in a refined one.

\begin{figure}[t!]
\centering
\includegraphics[width=.51\textwidth]{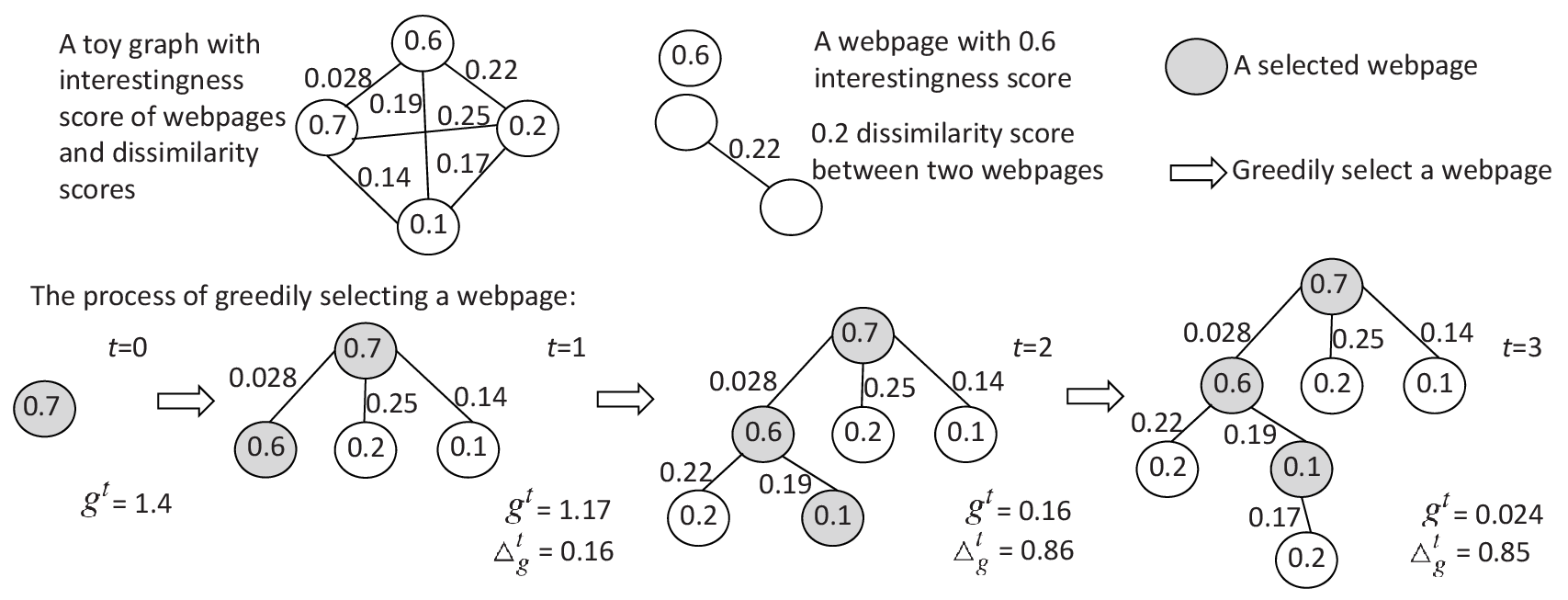}
\caption{A toy example illustrates Alg.~\ref{alg:optimize-goodness}.}
\label{fig:toyexample}
\end{figure}

\subsubsection{Determining the number of the selected webpages}

Considering that the size of different topics is very diverse, we present a heuristic approach based on the relative change of the goodness value in~\eqref{eqt:discrete-discrete} as follows:
\begin{equation}\label{eqt:threshold}
\triangle_g^t = \frac{g^t- g^{t+1}}{g^t},
\end{equation}
where $g^t$ is the maximal increase of the goodness value, $g^t=\max_i \{ \triangle(\mathbf{p}_i|\mathcal{P}) \}$ at the $t$-th selection. That is, $g^t$ measures the incremented goodness value when a webpage is added into a topic.
When the first noise webpage is being added into a topic at the $t{+}1$-th iteration, the value of $\triangle_g^t$ would be significantly increased. Because the goodness value $g^t$ is larger than that of $g^{t+1}$. Remark~\ref{rmk:boundedthreshold} proofs that $\triangle_g^t$ is bounded. To give a vivid illustration of Alg.~\ref{alg:optimize-goodness}, Fig.~\ref{fig:toyexample} use a 4-nodes toy graph to demonstrate how to greedily select the nodes by $\triangle_g^t$ at the $t$-th iteration.

\begin{myremark}\label{rmk:boundedthreshold}
The indicator $\triangle_g^t$ in~\eqref{eqt:threshold} is bounded, \emph{i.e.},  \mbox{$0\leq \triangle_g^t \leq 1$}.
\end{myremark}

Conventionally, we would select the webpages which have been selected before the maximal value $\hat{\triangle}_g^t$, that is, $\hat{\triangle}_g^t = \max_t\{\triangle_g^t\}$,($t = 1,\ldots, |\hat{C}|$). In this paper, we empirically use a margin $m$ ($m>0$) to increase the robustness of a topics. That is, the value $\triangle_g^t$ of a selected node should be larger than  $\hat{\triangle}_g^t+m$. Because the similarity between two webpages are easily corrupted by the sparse and noise textual feature. In this paper, we set $m = 0.1$.

{\color[RGB]{0,0,255} In summary, modeling interestingness utilizes pagerank~\eqref{eqt:stationarydistribution} to assign a high score to an interesting webpage; besides, these unimportant or error webpages introduced in the bundling step~\eqref{eqt:bundle} would be refined by submodular selection. Therefore, the error webpages are tended to be removed and barely influence the topic ranking results.} As Fig.~\ref{fig:toyexample} illustrated, by selecting the maximal value points, we would select two nodes with the interestingness scores with 0.7 and 0.6 as a refined topic.

\textbf{Time Complexity:} Given a list of topic candidates \mbox{$\mathcal{C}= C_0 \prec, \ldots, \prec C_K $},  the proposed BR method is summarized in Alg.~\ref{alg:AR3}. The time complexity of Alg.~\ref{alg:AR3} consists of both the bundling process and the refining one.

Concretely, \textbf{the bundling process} leads to a complexity of $O(K{\cdot}W)$ ($W{\ll} K$); while, the time cost of \textbf{the refining process} is $O(K(T\cdot nz(P) + |\widehat{C}|\cdot M))$. Usually, the number size of an unrefined topic is larger than $M$ (\emph{i.e.}, $M<|\widehat{C}|$). Therefore, the time complexity of the proposed BR is $O(K(T\cdot nz(P) + |\widehat{C}|\cdot M+W))$. {\color[RGB]{0,0,255} In practice, $T$ is smaller than 20, the maximum number of a coarse topic $M$ ($|\widehat{C}|> M$) is also smaller than 50, and the matrix $P$ is very sparse. Therefore, the proposed BR is quite efficient.}

\begin{algorithm}[t!]
\small
{
\SetAlgoLined
\caption{Refining a coarse topic}\label{alg:optimize-goodness}
\KwIn{The dissimilarity matrix $D\in\mathbb{R}^{|\widehat{C}|\times |\widehat{C}|}$, the weight $\lambda\geq 2$, and the score vector $\mathbf{\triangle_g} \in \mathbb{R}^{|\widehat{C}|\times1}$\;}
\textbf{Initialize:} $\mathcal{P}\leftarrow \varnothing $, and $\mathbf{\triangle_g} \leftarrow 1e-4$\;
\For{$t=0$ to $|\hat{C}|$}
{
Find $j=\arg\max_{i}(\triangle(\textbf{p}_i|\mathcal{P})|i=1,\ldots,N; i\notin \mathcal{P})$\;
Add $\mathbf{x}_j$ to $\mathcal{P}$\;
Update the value $\triangle_g^t$ by~\eqref{eqt:threshold}\;
}
Sort values $\triangle_g^t$ in a descended order, and find the index $\hat{t}$ with the maximal value, $\hat{t}=\max_t \triangle_g^t$\;
Group webpages $\textbf{p}_i$ ($i= 0,\ldots, \hat{t}$) into a refined topic\;
\KwOut{A refined topic.}
}
\end{algorithm}

\begin{algorithm}[t!]
\small
{
\SetAlgoLined
\caption{The proposed BR}\label{alg:AR3}
\KwIn{A list of topic candidates $\mathcal{C}=\{C_k\},k=1,\ldots, K$\;}
\textbf{Initialize:} $W \leftarrow 100$, $\tau \leftarrow 0.4$, and \mbox{$th\leftarrow$ $k$-ACV}\;
Compute the interestingness score $i_k$ by~\eqref{eqt:poissondeconvolution}\;

/* the \emph{Bundling} step */\\
\For {$k=0$ to $K$}
{
$\widehat{C}_k\leftarrow C_k$\;
\If {$J_{kl}$ $_{(l=k+1,\ldots,k+W)}$ $\geq \tau$}
{
 $\widehat{C}_k \leftarrow \widehat{C}_k \cup C_l$\;
}
Replace $C_k$ with $\widehat{C}_k$ in the list $\mathcal{C}$\;
}
Remove redundance in $\{\widehat{C}\}$ by NMS\;

/* the \emph{Refining} step */\\
\For {$k=0$ to $K$}
{
Build a graph from a coarse topic $\widehat{C}_k$\;
Compute the interestingness of webpages by~\eqref{eqt:stationarydistribution}\;
Refine the coarse topic $\widehat{C}_k$ by Alg.~\ref{alg:optimize-goodness}\;
}
\KwOut{A set of refined topics.}
}
\end{algorithm}

\section{Experiment}\label{sec:results}

\subsection{Datasets, Features and Evaluation Criteria}

\textbf{Dataset:} We evaluate our method on two public data sets, \emph{i.e.}, MCG-WEBV~\cite{Cao-Ngo-Zhang-Li-2011} and YKS~\cite{Zhang-Li-Chu-Wang-Zhang-Huang-2013}. {\color[RGB]{0,0,255} To our best knowledge, there are no publicly accessible web topic dataset with nearly 96\% noisy webpages which are not belongs to any topics. Both MCG-WEBV and YKS are very challenge to the approaches for clustering in a sea of noises. }
MCG-WEBV is built from the ``Most viewed'' videos of ``This month'' on YouTube from Dec. 2008 to Feb. 2009. For MCG-WEBV, the surrounding titles and the comments about the videos are considered as a set of words. YKS is a cross-media data set, where the meta data of YKS contains news articles on \emph{Sina} and the titles, the tags and the descriptions of web videos on \emph{YouKu}, from May 2012 to June 2012. The statistics of two data sets are summarized in Table~\ref{DatasetDetail}.

\begin{table}[t!]
\begin{tiny}
\caption{A summarization of MCG-WEBV and YKS.}
\label{DatasetDetail}
\begin{tabular}{|m{0.04\textwidth}<{\centering}|c|c|c|c|m{0.16\textwidth}<{\centering}|}
\hline
Dataset                                & \#Topic & \#Webpages  & \#Dictionary & \#Images  & The modalities used in the experiment                                                                                        \\ \hline
MCG-WEBV\cite{Cao-Ngo-Zhang-Li-2011} & 73      & 3,660       & 9,212        & 108,925 &  \begin{tabular}[c]{@{}l@{}}1. Keyframes of video clips;\\ 2. Titles, tags and descriptions.\end{tabular}                    \\ \hline
YKS\cite{Zhang-Li-Chu-Wang-Zhang-Huang-2013}                   & 318     & 18,399      & 54,620       & 71,063  &  \begin{tabular}[c]{@{}l@{}}1. Keyframes of web videos on Youku;\\ 2. Articles,titles and tags of news on Sina.\end{tabular} \\ \hline
\end{tabular}
\end{tiny}
\end{table}

Noise and sparsity are observed in two data sets. For instance, the sizes of dictionaries of both data sets are extremely large, \emph{e.g.}, 80,294 for YKS. Because the dictionaries of two sets contain not only the multi-language words but also user-defined abbreviations. Moreover, the average number of words in a webpage is also extremely small, \emph{e.g.}, 35 for MCG-WEBV and 228 for YKS. As a result, the textual features from social media are more noisier and shorter than news articles~\cite{zhao-ecir-2011}.

\textbf{Features:} In the pre-processing stage, YKS is tokenized by \emph{NLTK}\footnote{www.nltk.org} package. TF-IDF is used to encode the textual feature, and Fisher Vector (FV)~\cite{Sanchez-ijcv-2013} is used to encode the keyframes of video clips. Once keyframes are encoded, we use video signature~\cite{cheung-video-signature-tcsvt-2003} to compute the similarity between two video clips. The cosine distance is used to measure the similarity between the textual features. The $k$ of $k$-$N^2SG$s are assigned as 100 and 10 for both the textual graph and the visual one on MCG-WEBV, respectively. Since YKS is more noisy than that of MCG-WEBV, the $k$ of $k$-$N^2SG$s are set to 20 and 10 for the textual graph and the visual one, respectively.

During the experiments, SC~\cite{pang2013unsupervised} is assigned with a set of thresholds, $\{0.1,0.5,0.9\}$, and Non-negative Matrix Factorization with Random walk (NMFR)~\cite{yang2012clustering} is used to generate topics. The number of clusters for each threshold are assigned as $\{100,500,900,1300\}$. In the experiments, we have generated 4,240 topics for MCG-WEBV and 5,252 ones for YKS, respectively.

\textbf{Evaluation Criteria:} Top-10 $F_1$ versus Number of Detected Topics (NDT), and accuracy versus False Positive Per Topic (FPPT)~\cite{pang2013unsupervised} are used.


For top-10 $F_1$ versus NDT, if a detected topic is matched with the ground truth, {\color[RGB]{0,0,255} the averaged F1 scores of the top 10 detected topics} are averaged to measure the performance:
\begin{equation}\label{eqt:f1}
F_1=\frac{2\times Precision \times Recall}{Precision+Recall},
\end{equation}
where $ D$ is a detected topic, $G$ is a ground truth topic, and $|\cdot |$ denotes the number of the webpages in a topic. Therefore, $Precision=\frac{|D \cap G|}{|D|}$ is the precision, $Recall=\frac{|D \cap G|}{|G|}$ is the recall. However, top-10 $F_1$ versus NDT does not consider the influence of the number of detected topics. Therefore, accuracy versus FPPT is further proposed to handle this problem~\cite{pang2013unsupervised}.


For accuracy versus FPPT, if a topic is correctly detected, how many falsely detected topics are generated by a detection system, where accuracy is defined as follows:
\begin{equation}\label{equa:accuracy}
\text{Accuracy}=\frac{\# \text{Successful} }{\# \text{Groundtruth}}.
\end{equation}
A topic candidate $D$ is recognized as a successful detection, if Normalized Intersected Ratio (NIR) $r=\frac{|D \cap G|}{|D \cup G|}$ is larger than a threshold. Following the previous work~\cite{pang2013unsupervised}, NIR with a threshold of 0.5 is used in our experiments.
{\color[RGB]{0,0,255} In summary, the top-10 F1 focuses on measuring the partial performance of the top ranked topics. While the accuracy measures the overall performance of the detected results. }


\subsection{Methods in Comparison Study}

Our experimental goal is to compare the proposed approach with four state-of-the-art methods:

\begin{itemize}
\item [1.]\textbf{Event-Clustering Based Method (ECBM)~\cite{Cao-Ngo-Zhang-Li-2011}.} Different from the PD-based method~\cite{pang2013unsupervised}, this work~\cite{Cao-Ngo-Zhang-Li-2011} first clusters the tags from each time slice, and then both the Near Duplicated Keyframes (NDKs) and the tag-based clusters are grouped into topics. Note that this approach involves many engineering details and the tuning of the hyper-parameters. We implemented this method by ourselves and reported the best tuned results.
\item [2.]\textbf{Multi-Modality Graph (MMG)~\cite{Zhang-Li-Chu-Wang-Zhang-Huang-2013}.} This baseline belongs to the similarity-graph based method. Zhang \emph{et al.}~\cite{Zhang-Li-Chu-Wang-Zhang-Huang-2013} utilize the NDKs of video clips and the textual information to build the similarity graph~\cite{papadopoulos-11}, where graph shift (GS)~\cite{liu-icml-2009} is used to discover web topics. This method mainly depends on the noise-resistant power of GS. The comparison between MMG and our method illustrates that without designing any noise-resistant component, putting fragments into a whole is a promising approach to generate complete topics.
\item [3.]\textbf{Maximal Cliques with Poisson Deconvolution (MCPD)~\cite{pang2013unsupervised}.} Rather than using NMFR to generate topics, MCPD leverages Maximal Cliques (MCs) to generate topics. This comparison demonstrates the generalization ability of the proposed BR method across different topic patterns. Because a good generalization ability guarantees the success of the combination of  the BR method and the other topic patterns, \emph{e.g.}, spectral clustering~\cite{shi-malik-spectralcluster-pami2000}.
\item [4.]~\textbf{Latent Poisson Deconvolution (LPD)~\cite{pang-tao-lpd-icme-2016}.} This method achieves the state-of-the-art performances on both \mbox{MGC-WEBV} and YKS. Note that rather than using multiple graphs in LPD, our method only utilizes one graph to rank topics. This demonstrates that our approach can meet or even surpass the state-of-the art method without exploiting the multiple graphs.
\end{itemize}

\subsection{Analysis of Our Approach}\label{sec:sub:analysisAR3}

In this Section, we use MCG-WEBV to test the effectiveness of our method as follows:
\begin{itemize}
\item Verify the effectiveness of each component in BR;
\item The generalization ability of BR across different topic patterns;
\end{itemize}


\begin{figure}[t!]
\centering
\includegraphics[width=.45\textwidth]{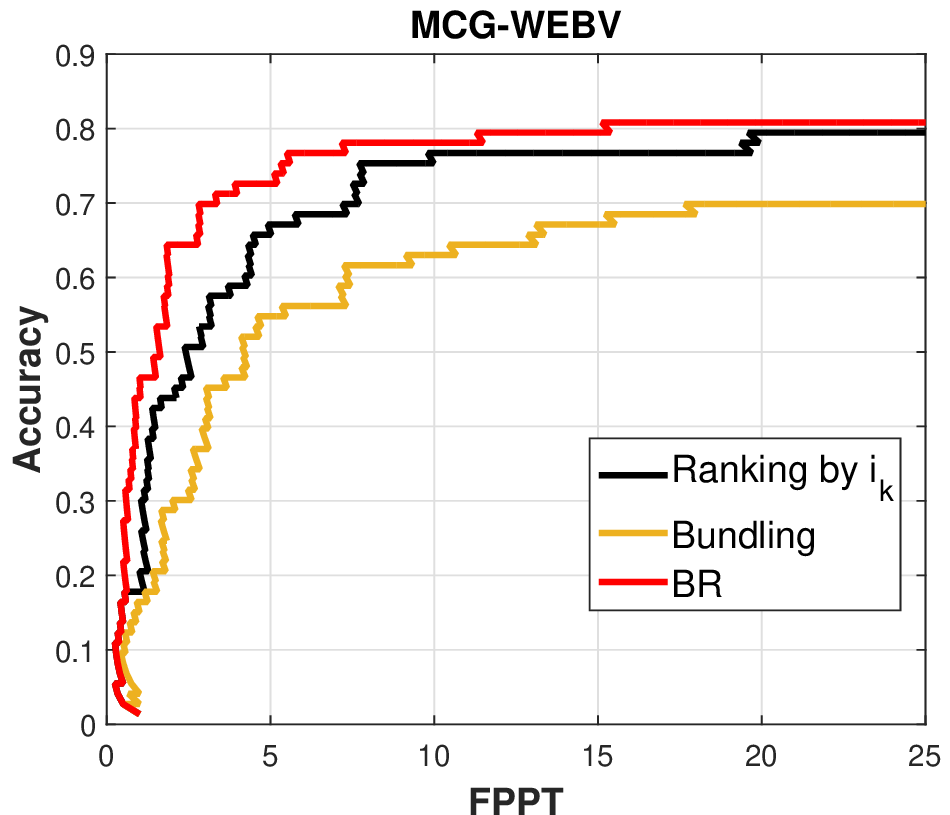}
\caption{Effectiveness of each component of BR on MCG-WEBV (best viewed in color).}
\label{fig:effectiveness_steps}
\end{figure}

\subsubsection{The Effectiveness of Each Component in BR}

Fig.~\ref{fig:effectiveness_steps} compares the effectiveness of each component of BR \emph{i.e.}, \emph{Bundling} and the combination of \emph{Bundling} and \emph{Refining} in Alg.~\ref{alg:AR3}.

In this experiment, a baseline method~\cite{pang2013unsupervised} (we term it as ``Ranking by $i_k$'') ranks topics by interestingness $i_k$. Fig.~\ref{fig:effectiveness_steps} shows that the proposed BR significantly surpasses the results of ``Ranking by $i_k$''. The comparisons demonstrate the advantages of BR in both removing fragment topics and mining more hot topics.

Moreover, Fig.~\ref{fig:effectiveness_steps} verifies the importance of the refining component in BR. The curve with the legend ``Bundling'' achieves the worst performance in Fig.~\ref{fig:effectiveness_steps}. Because some hot topics are evolved into the inaccurate ones due to the absorbed noise webpages. As expected, compared with the curve with the legend ``Bundling'', the proposed BR significantly improves the accuracy of a detection system. For example, when FPPT is equal to 5, the BR outperforms the ``Bundling'' by 20\% accuracy nearly.

In summary, although the \emph{Bundling} component in the BR puts fragments into a whole, the accuracy of coarse topics are also significantly reduced. To remedy the drawback of the \emph{Bundling} component, the~\emph{Refining} component of BR strips noise webpages from coarse topics, and brings the gain in accuracy. Therefore, both two components are very important to guarantee the performance of the proposed method.

\begin{figure}[t!]
\centering
\includegraphics[width=.45\textwidth]{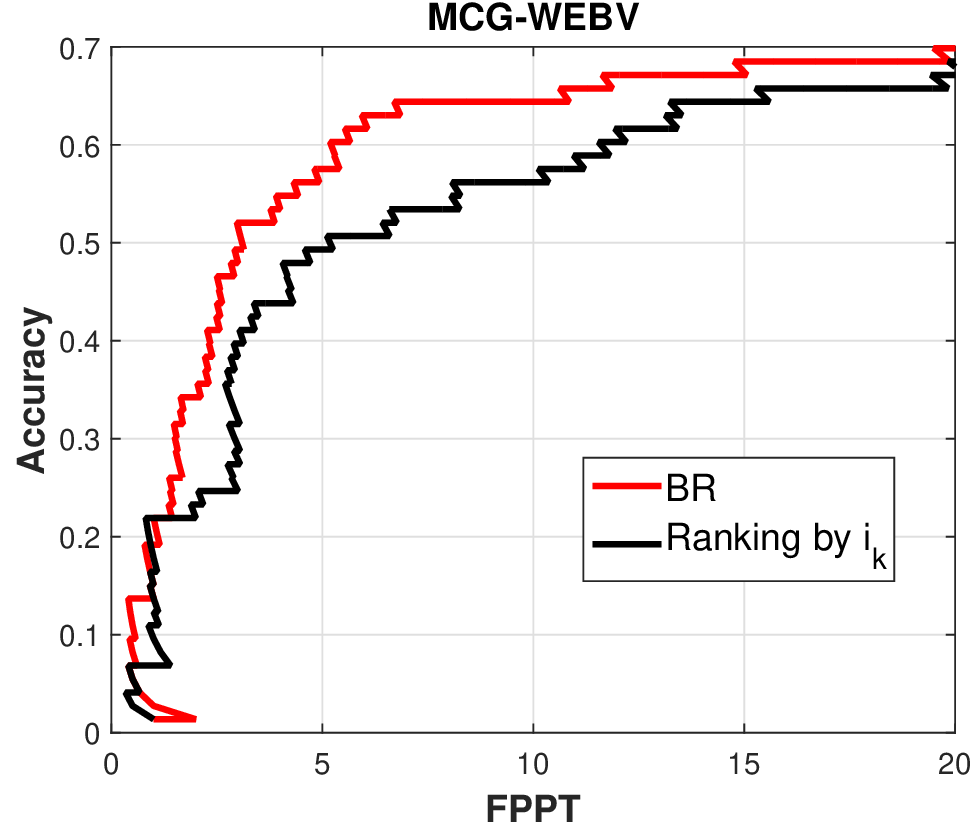}
\caption{The accuracy versus FPPT curves for the combination of MCPD and BR on MCG-WEBV (best viewed in color).}
\label{fig:generalizationoftopic_patterns}
\end{figure}

\subsubsection{The Generalization Ability of BR Across Different Types of Topics}

Fig.~\ref{fig:generalizationoftopic_patterns} illustrates the effectiveness of BR on MC-based topic pattern~\cite{pang2013unsupervised}. Concretely,
the proposed BR achieves significantly improved performance for this type of topic pattern. Note that MC, a connectable pattern, is totally different from the clusters generated by NMFR. Therefore, BR is expected to have a good generalization ability across the different topic patterns.

The generalization ability empowers the proposed BR easily extendable to other efficient topic patterns, \emph{e.g.}, spectral clustering.

%
%

\subsection{Qualitative Comparisons with State-Of-The-Art Methods}

To make the comparisons as meaningfully as possible, we use the same experimental protocols proposed by each data set.

\subsubsection{\textbf{Web-Video Topic Detection in MCG-WEBV}}

Fig.~\ref{fig:top10_f1onmcg-webv} illustrates the comparison results by the top-10 $F_1$ \emph{v.s.} NDT on MCG-WEBV. The proposed BR achieves the highest top-10 $F_1$ score than the other state-of-the-art methods. Besides, the top-10 $F_1$ scores of the BR method increase quickly with respect to the number of the generated topics. For instance, to achieve approximate 0.9 top-10 $F_1$ score, MCPD~\cite{pang2013unsupervised}, and MMG~\cite{Zhang-Li-Chu-Wang-Zhang-Huang-2013} generate 70, and 179 topics respectively, while our method only generates 20 topics. It means that a user of BR could quickly find top 10 interesting hot topics without screening more topics.

\begin{figure}[t!]
\centering
\includegraphics[width=.445\textwidth]{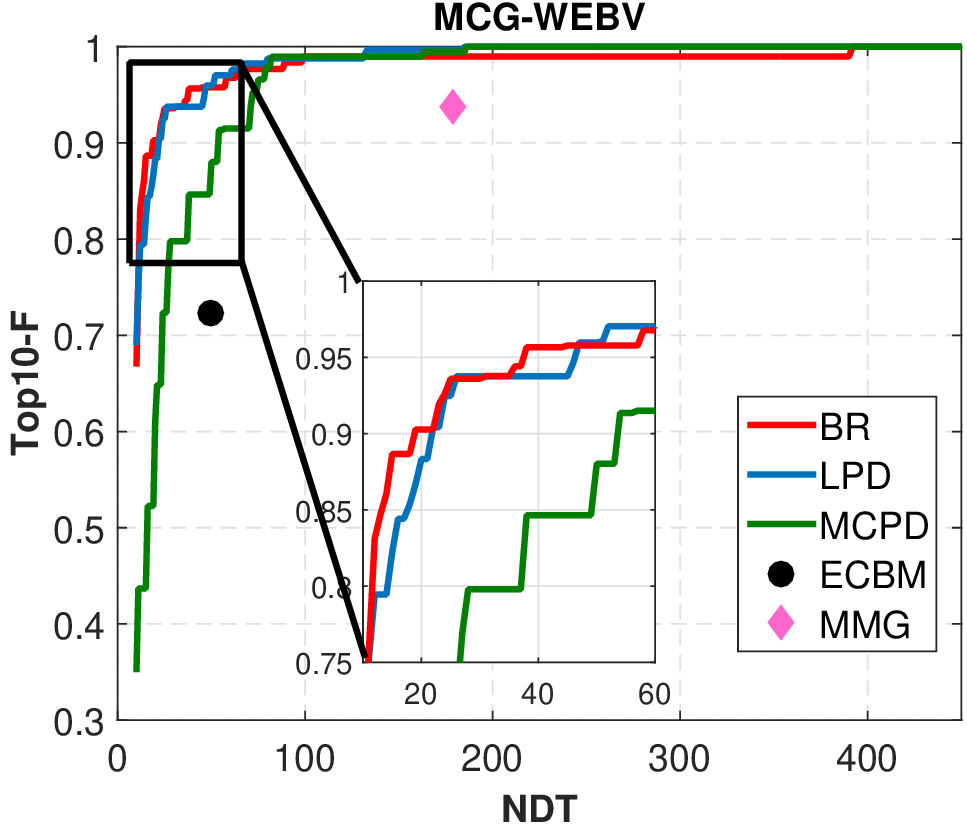}
\caption{Comparisons between the state-of-the-art methods and our method by Top-10 $F_1$ versus NDT on MCG-WEBV (best viewed in color).}
\label{fig:top10_f1onmcg-webv}
\end{figure}

In Fig.~\ref{fig:top10_f1onmcg-webv}, BR efficiently boosts the performance of the PD-based methods, even achieving a slightly better result than LPD~\cite{pang-tao-lpd-icme-2016}. The comparison between LPD and BR in Fig.~\ref{fig:top10_f1onmcg-webv} discovers two important observations: (1) efficiently removing fragment topics is very critical to web topic detection, since the number of hot topics is very small; (2) the~\emph{Refining} component in BR indeed improves the $F_1$ scores of hot topics via removing noise webpages.

\begin{figure}[h!]
\centering
\includegraphics[width=.45\textwidth]{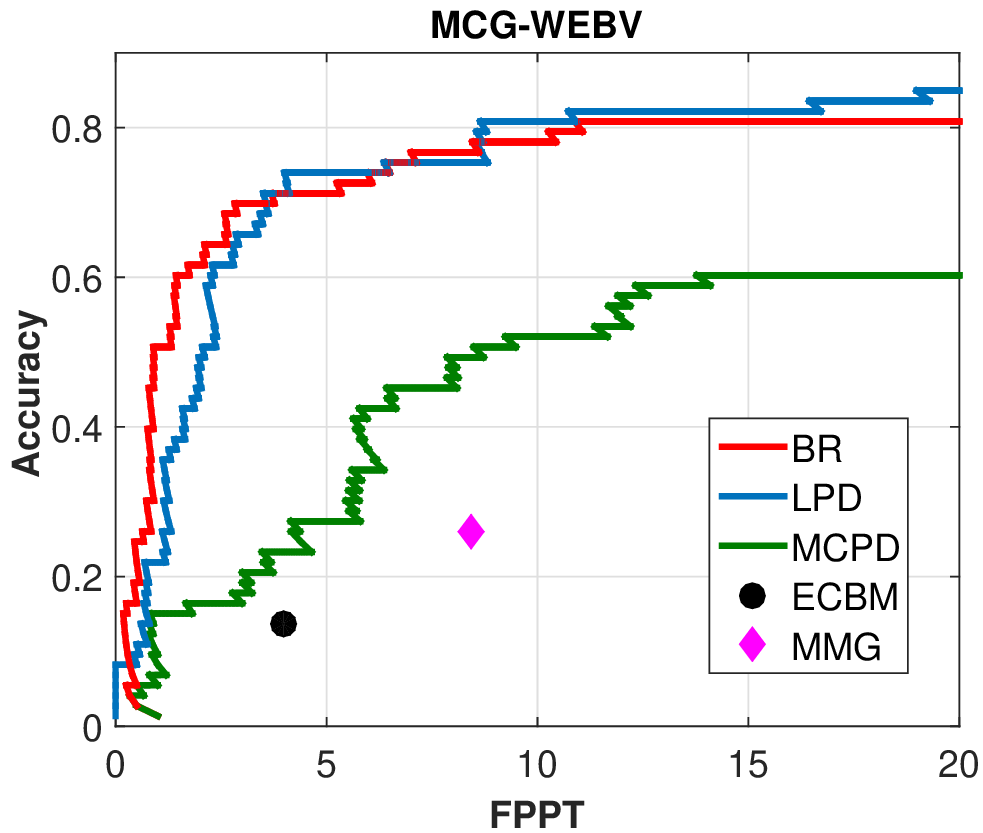}
\caption{Comparisons between the state-of-the-art methods and our method by accuracy versus FPPT on MCG-WEBV (best viewed in color).}
\label{fig:fppt_on_mcg-webv}
\end{figure}

To further evaluate the topic-wise performance, the accuracy \emph{v.s.} FPPT curves are evaluated. As illustrated in Fig.~\ref{fig:fppt_on_mcg-webv}, our approach is consistently better than MCPD~\cite{pang2013unsupervised}, MMG~\cite{Zhang-Li-Chu-Wang-Zhang-Huang-2013}, and ECBM~\cite{Cao-Ngo-Zhang-Li-2011}. For instance, when FPPT equals to 5, the accuracy of our method is 0.76; while the accuracies of MMG, ECBM and MCPD are all less than 0.3. The explanation indicates that BR is very efficient to boost the performances, as a conclusion in Section~\ref{sec:sub:analysisAR3}.

Unexpectedly, without the complex ranking method in LPD~\cite{pang-tao-lpd-icme-2016}, BR uses the same topic pattern (\emph{i.e.}, NMFR~\cite{yang2012clustering}) achieves better accuracies than LPD in Fig.~\ref{fig:fppt_on_mcg-webv}. Especially, when FPPT ranges from 0 and 5, BR significantly outperforms LPD. This indicates that both efficiently removing uninteresting topics and bundling fragments into a whole is critical than the complex topic ranking method.

\begin{figure}[h!]
\centering
\includegraphics[width=.45\textwidth]{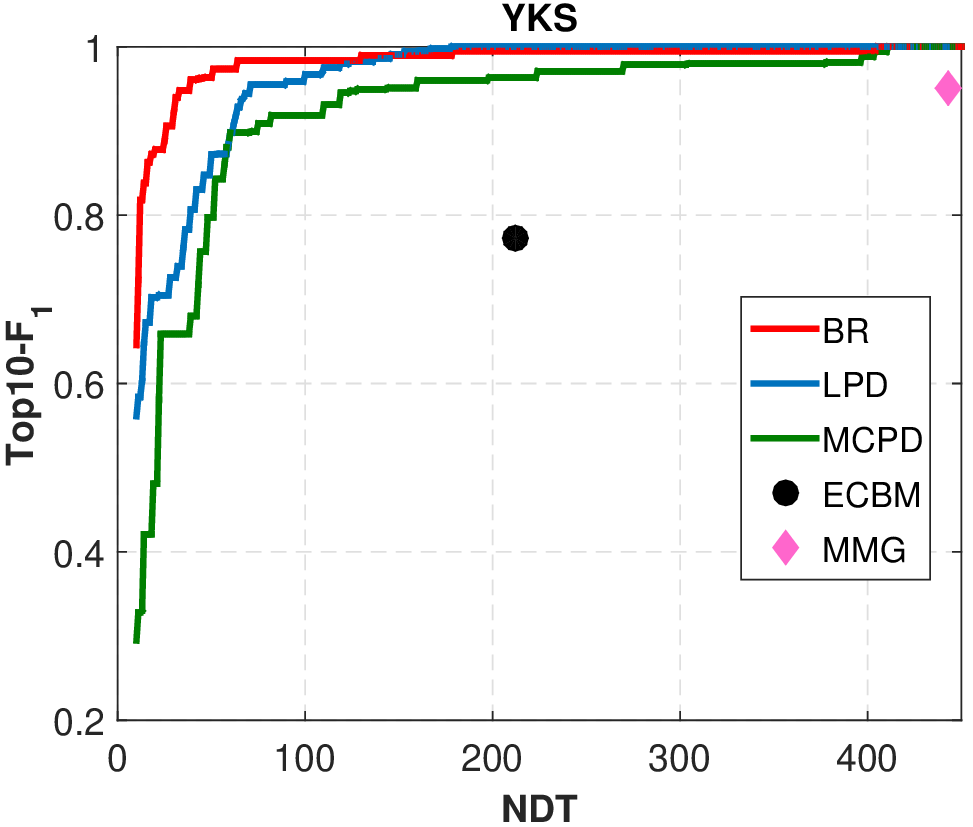}
\caption{Comparisons between the state-of-the-art methods and our method by Top-10 $F_1$ versus NDT on YKS (best viewed in color).}
\label{fig:topf1-vs-ndt-on-yks}
\end{figure}

\begin{table*}[t!]
\centering
\footnotesize{
\caption{COMPARISONS OF RUNNING TIME (IN SECONDS) AMONG DIFFERENT METHODS (millisecond).} \label{tab:runningtime}
\begin{tabular}{|c|c|c|c|c||c|c|}
\hline
\multirow{2}*{Dataset} &  \multirow{2}*{ECBM} &  \multirow{2}*{MMG} &   \multirow{2}*{MCPD} &  \multirow{2}*{LPD} & \multicolumn{2}{c|}{BR} \cr\cline{6-7}
& &    & &  & Ranking+Bundling & Refining   \cr \hline\hline
 MCG-WEBV & 146 & \textbf{15} & 69,120  & 6,297 & 208  & 17\cr\hline
 YKS & 434 & \textbf{252} & 95,040  & 17,474  &  258 & 62 \cr \hline
\end{tabular}
}
\end{table*}

\subsubsection{\textbf{Web Topic Detection in YKS}}

YKS, a cross-platform data set, requires to grasp more diverse types of the topic patterns than MCG-WEBV. Fig.~\ref{fig:topf1-vs-ndt-on-yks} shows that our method consistently outperforms MCPD~\cite{pang2013unsupervised}, MMG~\cite{Zhang-Li-Chu-Wang-Zhang-Huang-2013}, ECBM~\cite{Cao-Ngo-Zhang-Li-2011}, and LPD~\cite{pang-tao-lpd-icme-2016}. For instance, both LPD~\cite{pang-tao-lpd-icme-2016} and MCPD~\cite{pang2013unsupervised} have to generate approximate 120 topics and about 400 ones in order to archive 1 top-10 $F_1$, respectively. In contrast, our approach only generates about 70 topics to achieve the same top-10 $F_1$ score. Obviously the improved results indicate that BR accurately refines topics in the \emph{Refining} component.

\begin{figure}[h!]
\centering
\includegraphics[width=.45\textwidth]{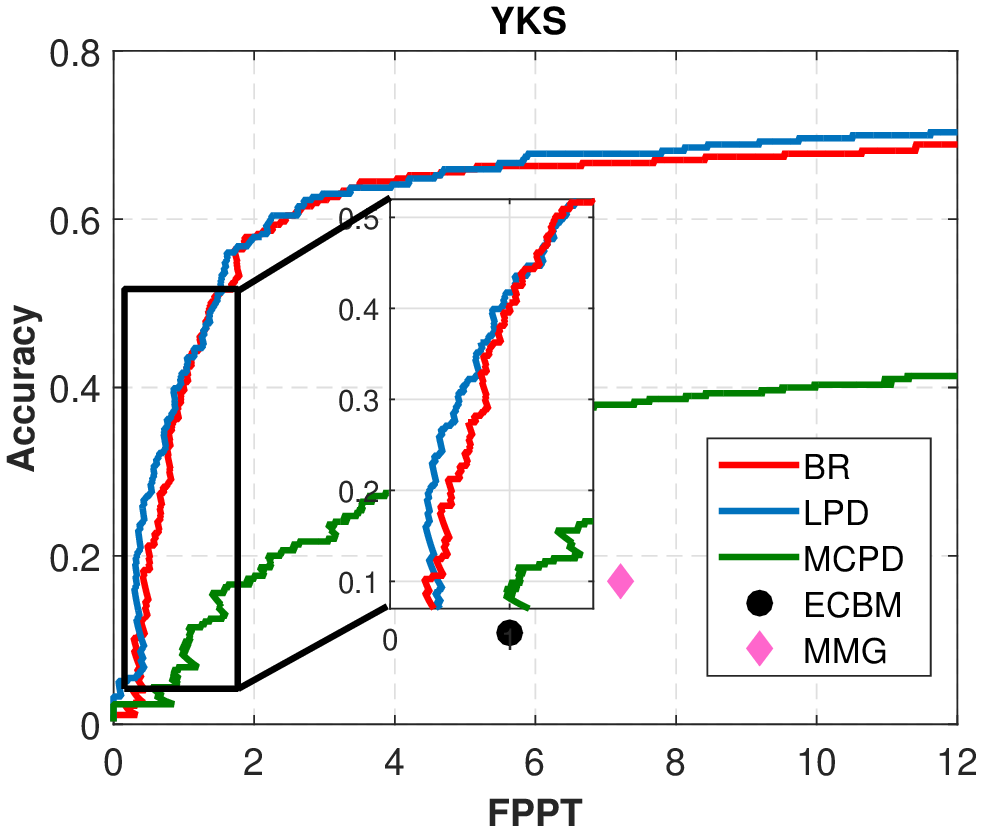}
\caption{Comparisons between the state-of-the-art methods and our method by accuracy versus FPPT on YKS (best viewed in color).}
\label{fig:accuracy-vs-fppt-on-yks}
\end{figure}

Fig.~\ref{fig:accuracy-vs-fppt-on-yks} further illustrates the accuracy versus FPPT curves on YKS. Our approach consistently outperforms MCPD~\cite{pang2013unsupervised}, MMG~\cite{Zhang-Li-Chu-Wang-Zhang-Huang-2013}, and ECBM~\cite{Cao-Ngo-Zhang-Li-2011}, and is comparable to the state-of-the-art method LPD~\cite{pang-tao-lpd-icme-2016}. Although without adopting the complex ranking method as LPD, BR achieves very similar results to LPD. However, when FPPT is equal to 12, BR achieves a slightly worse performance than LPD. It indicates that a very small portion of the topics has been wrongly refined.

\subsection{Efficiency}

To give readers an intuition about the efficiency, we report the empirical running time for different methods. Our experiments are conducted on a server (Windows) with Intel Core
CPU with 3.2 GHz and 32 GB main memory. The source code was implemented with Matlab.
Table~\ref{tab:runningtime} compares that the running time among each methods.

Although MMG~\cite{Zhang-Li-Chu-Wang-Zhang-Huang-2013} achieves the lowest value of running time, MMG are worse than the PD-based methods (\emph{i.e.}, MCPD~\cite{pang2013unsupervised}, LPD~\cite{pang-tao-lpd-icme-2016} and BR), in terms of both the accuracy and Top-10 $F_1$.

During the generation of topic candidates, LPD almost costs \mbox{2$\times$} time than that of BR, since LPD utilizes two graphes to generate topic candidates while BR just uses one graph; besides, ranking topics in LPD costs about \mbox{27.9$\times$} on MCG-WEBV (about \mbox{54.6$\times$} on YKS) time than the combination of Ranking, Bundling and Refining in BR. On the other side, in terms of effectiveness, BR surpasses LPD on MCG-WEBV, and meets LPD on YKS. In summary, the proposed BR not only at least meets LPD in terms of effectiveness, but also significantly overpasses LPD in terms of efficiency.

\section{Conclusion}\label{sec:conclusion}

In this paper, BR is proposed to boost performances of the DBR approach. BR consists of the \emph{Bundling} and the \emph{Refining} steps. The \emph{Bundling} step puts fragment topics into a whole yet coarse topic; while, the \emph{Refining} step refines a coarse topic into a fine one. Considering than the \emph{Refining} step is equivalent to strip noise webpages, without the complex training process, the proposed submodular selection adaptively refines coarse topics into more accurate, compact and coherent ones. Experiments show that the proposed method meets or surpasses the state-of-the-art methods in terms of both effectiveness and efficiency.

The promising results of this paper motivate a further examination of the components of the proposed BR. A high-efficient method to generate topic candidates should be investigated. For instance, an incremental method to generate topic candidates on a graph~\cite{Held2016DynamicCI} may remedy the drawback of the off-line clustering method here.

\section{Appendix}

The proof of Remark~\ref{rmk:boundedthreshold} is as follows:
\begin{proof}
Let $g_\mathbf{m}^t$ denotes the value of~\eqref{eqt:discrete-discrete} which selects the $\mathbf{m}$ webpage at the $t$-th selection, and $g_\mathbf{n}^{t+1}$ be the value of the selection of the $\mathbf{n}$ webpage, we have:
\begin{equation*}
\begin{split}
g_\mathbf{m}^t &= \max_\mathbf{m} \left( \triangle(\mathbf{p}_\mathbf{m}|\mathcal{P}) \right) \\
&= \lambda \pi_\mathbf{m}- \left( \sum_{i \in \mathcal{P}}\pi_i D_{i\mathbf{m}} \pi_\mathbf{m} + \sum_{j \in \mathcal{P}}\pi_\mathbf{m} D_{\mathbf{m}j} \pi_j \right)\\
& \geq \lambda \pi_\mathbf{n}- \left( \sum_{i \in \mathcal{P}}\pi_i D_{i\mathbf{n}} \pi_\mathbf{n} + \sum_{j \in \mathcal{P}}\pi_\mathbf{n} D_{\mathbf{n}j} \pi_j \right)\\
&\geq \lambda \pi_\mathbf{n}- \left( \sum_{i \in \mathcal{P} \cup \{\mathbf{m}\}}\pi_i D_{i\mathbf{n}} \pi_\mathbf{n} + \sum_{j \in \mathcal{P} \cup \{ \mathbf{m}\} }\pi_\mathbf{n} D_{\mathbf{n}j} \pi_j \right)\\
& = \max_\mathbf{n} \left( \triangle(\mathbf{p}_\mathbf{n}|\mathcal{P} \cup \{\mathbf{m}\} ) \right) = g_\mathbf{n}^{t+1}
\end{split}
\end{equation*}

By above inequality and the condition $g_\mathbf{m}^t> 0$, we have:
 \begin{equation*}
 \begin{split}
 \because \   \ & 0 \leq  g_\mathbf{m}^t - g_\mathbf{n}^{t+1} \leq g_\mathbf{m}^t, \\
 \therefore \   \ & 0 \leq  \frac{g_\mathbf{m}^t - g_\mathbf{n}^{t+1}}{g_\mathbf{m}^t}  \leq1.
\end{split}
\end{equation*}
\end{proof}

\section*{References}
\bibliography{sigproc,tmmtopic,icme}

\end{document}